\documentclass[12pt,reqno]{amsart}
\usepackage{amsmath}
\usepackage{latexsym}
\usepackage{amsfonts}
\usepackage{amssymb}
\usepackage{color}
\usepackage{bbm,dsfont}
\usepackage{graphicx}
\usepackage{enumerate}

\usepackage{tikz}
\usepackage[caption=false]{subfig}

\newtheorem{proposition}{Proposition}
\newtheorem{theorem}{Theorem}
\newtheorem{lemma}{Lemma}

\theoremstyle{definition}

\usepackage{array}
\usepackage{booktabs,tabularx}
\usepackage{multirow}
\usepackage[graphicx]{realboxes}

\definecolor{darkgreen}{rgb}{0,0.6,0.2}

\newcommand{\Rb}{\mathbb R} 
\newcommand{\Cb}{\mathbb C} 
\newcommand{\Zb}{\mathbb Z} 
\newcommand{\abs}[1]{\left| #1 \right|} 

\newcommand{\hh}{\mathcal{H}} 
\newcommand{\lh}{\mathcal{L(H)}} 
\newcommand{\ip}[2]{\left\langle\,#1\,|\,#2\,\right\rangle} 
\newcommand{\kb}[2]{|#1\rangle\langle#2|} 
\newcommand{\no}[1]{\left\|#1\right\|} 
\newcommand{\tr}[1]{\textrm{tr}\left[#1\right]} 
\newcommand{\ran}{\textrm{Im}\,} 
\newcommand{\rank}{{\rm rank}\,} 
\newcommand{\id}{\mathbbm{1}} 
\newcommand{\spanno}[1]{{\rm span}\,\left\{ #1 \right\}}

\newcommand{\va}{\mathbf{a}} 
\newcommand{\vb}{\mathbf{b}} 
\newcommand{\vsigma}{\boldsymbol{\sigma}} 
\newcommand{\vnull}{\mathbf{0}}

\newcommand{\Ao}{\mathsf{A}}
\newcommand{\Bo}{\mathsf{B}}
\newcommand{\Co}{\mathsf{C}}
\newcommand{\Uo}{\mathsf{U}}

\newcommand{\ca}[1]{{\mathcal{#1}}}
\newcommand{\Lam}{\Lambda}
\newcommand{\lam}{\lambda}

\newcommand{\phii}{\varphi}
\newcommand{\bgam}{{\boldsymbol\gamma}}
\newcommand{\rme}{{\rm e}}
\newcommand{\rmi}{{\rm i}}
\newcommand{\ext}[1]{{\rm ext}\,#1}
\newcommand{\OO}[1]{\ca{O}(#1)}
\newcommand{\cOO}{\ca{CO}(\Omega)}
\newcommand{\pair}[2]{\left\langle\,#1\,,\,#2\,\right\rangle} 
\newcommand{\aligno}[1]{\begin{align*} #1 \end{align*}}

\newcommand{\equano}[1]{\begin{equation*} #1 \end{equation*}}
\newcommand{\equasi}[2]{\begin{equation} \label{#1} #2 \end{equation}}

\begin{document}

\title[Extremal compatible observables by means of two MUB]{Constructing extremal compatible quantum observables by means of two mutually unbiased bases}

\dedicatory{{Dedicated to the memory of Paul Busch}}

\author[]{Claudio Carmeli}
\address{\textbf{Claudio Carmeli}; DIME, Universit\`a di Genova, Via Magliotto 2, I-17100 Savona, Italy}
 \email{carmeli@dime.unige.it}

\author[]{Gianni Cassinelli}
\address{\textbf{Gianni Cassinelli}; Dipartimento di Fisica, Universit\`a di Genova, Via Dodecaneso 33, I-16146 Genova, Italy}
\email{cassinelli@ge.infn.it}
          
\author[]{Alessandro Toigo}
\address{\textbf{Alessandro Toigo}; Dipartimento di Matematica, Politecnico di Milano, Piazza Leonardo da Vinci 32, I-20133 Milano, Italy, and I.N.F.N., Sezione di Milano, Via Celoria 16, I-20133 Milano, Italy}
\email{alessandro.toigo@polimi.it}

\maketitle

\begin{abstract}
We describe a particular class of pairs of quantum observables which are extremal in the convex set of all pairs of compatible quantum observables. The pairs in this class are constructed as uniformly noisy versions of two mutually unbiased bases (MUB) with possibly different noise intensities affecting each basis. We show that not all pairs of MUB can be used in this construction, and we provide a criterion for determiniing those MUB that actually do yield extremal compatible observables. We apply our criterion to all pairs of Fourier conjugate MUB, and we prove that in this case extremality is achieved if and only if the quantum system Hilbert space is odd-dimensional. Remarkably, this fact is no longer true for general non-Fourier conjugate MUB, as we show in an example. Therefore, the presence or the absence of extremality is a concrete geometric manifestation of MUB inequivalence, that already materializes by comparing sets of no more than two bases at a time.
\end{abstract}

\section{Introduction}

Quantum incompatibility is one of the most striking and still elusive features of quantum theory \cite{HeMiZi16}. Basically, it can be condensed into the statement that, within the quantum world, not all observable quantities can be measured simultaneously. In a more precise reformulation, there exist pairs of quantum observables whose outcome probability distributions cannot be both statistically postprocessed from the results of a single joint measurement \cite{Holevo85,AlCaHeTo09}.

A very important observation circumventing the no-go essence of quantum incompatibility -- and dating back to the mid 70's -- asserts that even two incompatible quantum observables can be described by some joint measurement, if one contents himself with `unsharp' or `noisy' versions of the respective outcome statistics \cite{QTOS76,QDET76,PSAQT82}. Now, restricting the discussion to observables with a finite number of outcomes, a very simple way to generate smeared versions of their outcome probability distributions is by mixing them with the uniform -- i.e., white -- noise. The minimal amount of uniform noise yielding compatibility obviously depends upon the two original observables. In particular, it can be interpreted as an index quantifying their degree of incompatibility: the higher is the minimal noise, the more incompatible are the two observables. In more recent times, this has indeed been developed into a method for comparing the incompatibility content of different pairs of observables \cite{HeScToZi14,HeKiRe15} or even different probabilistic theories alternative to quantum mechanics \cite{BuHeScSt13,BaGaGhKa13,JePl17,Gu16}.

The above described way of turning two incompatible observables into a compatible pair is reminescent of the (random) robustness of entanglement for quantum states \cite{ViTa99}: in both cases, a classical feature -- separability for states or compatibility for pairs of observables -- is achieved by mixing the original entangled state or incompatible observables with a trivial state -- e.g., the maximally mixed state -- or trivial observable -- e.g., the uniform probability distribution -- respectively. By virtue of this analogy, the amount of uniform noise triggering compatibility is sometimes called (random) `noise robustness' of a pair of observables \cite{UoBuGuPe15,CaSk16,DeSkFrBr19}. The noise robustness has a clear geometric meaning: it is the linear coordinate of the point where the line segment joining the given pair of quantum observables and the pair of classical uniform probabilities pierces the convex compact set of all compatible pairs of observables \cite{Haapasalo15}.

Evaluating the noise robustness thus gives an inspection of the relative boundary of the set of all compatible pairs of observables. Since any convex compact set is completely determined by the knowledge of its extreme points, it is relevant to know for what kind of incompatible observables the boundary points found in this way are actually extremal compatible pairs. In the present paper, we answer this question for pairs of observables constituted by two mutually unbiased bases (MUB). That is, we completely characterize the pairs of MUB whose uniformly noisy versions actually yield extreme points. We show that in dimension $d\geq 3$, the property of yielding extreme points is determined by the eigenvalues of a $d^2\times d^2$ real symmetric matrix first introduced by Haagerup in \cite{Haagerup96} and canonically associated with the MUB at hand. In this way, we remarkably find that in some dimensions there are MUB both with and without the above mentioned property; extremality is then a concrete geometric manifestation of the existence of inequivalent pairs of MUB in sufficiently high dimensions.

The problem of determining the minimal amount of uniform noise which makes two MUB compatible was solved for the first time in \cite{CaHeTo12} under the assumption that the two bases are conjugate by the Fourier transform of a cyclic group. In \cite{DeSkFrBr19,UoLuMoHe16}, these results were extended to the case of arbitrary pairs of equally noisy MUB, and finally to arbitrary uniformly noisy MUB with possibly different and even negative noise parameters in \cite{CaHeToPRL19}. On the other hand, a general characterization of the extreme points of the set of all pairs of compatible observables is provided in \cite{GuCu18}. The results presented here are then essentially derived by combining \cite{CaHeToPRL19} and \cite{GuCu18}.

We finally outline the structure of the paper. In Sec.~\ref{sec:prel}, we establish the notations, formally state our problem and recall the main related results from \cite{CaHeToPRL19,Busch86}. In Sec.~\ref{sec:warm-up}, we preliminarly solve the problem in the simple two-dimensional case. In Sec.~\ref{sec:higher-d}, we extend the approach for dimension $2$ to arbitrary dimensions $d\geq 3$, and establish our first main result about the connection between extremality and the eigenvalues of the Haagerup matrix. Sec.~\ref{sec:examples} provides some applications: in Sec.~\ref{subsec:Fourier}, we focus on two Fourier conjugate MUB and prove that their uniformly noisy versions yield extremal pairs of compatible observables if and only if the dimension $d$ is odd; in the subsequent Sec.~\ref{subsec:numeric}, we characterize extremality for some non-Fourier conjugate pairs of MUB taken from \cite{TaZy06}. It turns out that in dimensions $d\geq 3$ there exists one isolated pair of uniformly noisy MUB that does not have any analogue in dimension $d=2$; we characterize the extremality of this pair in Sec.~\ref{sec:exceptional}. Our concluding Sec.~\ref{sec:discussion} contains a final discussion.

This paper is dedicated to Paul. We have learned from his books, his papers, his lectures, now we miss him as a friend, we miss his ideas and his collaboration. We can only remember Paul and offer this little contribution, hoping that he could have liked some ideas that come out also from his understanding of Quantum Mechanics.

\section{Preliminaries and notations}\label{sec:prel}

We fix a finite-dimensional complex Hilbert space $\hh$, with $\dim\hh = d$. We denote by $\lh$ the vector space of all linear operators on $\hh$, and we write $\id$ for the identity operator. An {\em observable} with outcomes in a finite set $\Omega$ is a map $\Ao:\Omega\to\lh$ such that $\Ao(x)$ is a positive operator for all $x\in\Omega$, and $\sum_{x\in\Omega} \Ao(x) = \id$. We write $\OO{\Omega}$ for the convex compact set of all observables with outcomes in $\Omega$. The uniform observable $\Uo_\Omega\in\OO{\Omega}$ is defined as $\Uo_\Omega(x) = (1/\abs{\Omega})\,\id$ for all $x\in\Omega$, where $\abs{\Omega}$ is the cardinality of $\Omega$.

Two observables $\Ao,\Bo\in\OO{\Omega}$ are {\em compatible} if there exists a third observable $\Co\in\OO{\Omega^2}$ such that its {\em margin observables}
\equano{
\Co_{[1]}(x) = \sum_{y\in\Omega} \Co(x,y)\,, \qquad\qquad \Co_{[2]}(y) = \sum_{x\in\Omega} \Co(x,y)
}
satisfy the equalities $\Co_{[1]} = \Ao$ and $\Co_{[2]} = \Bo$. In this case, $\Co$ is a {\em joint observable} of $\Ao$ and $\Bo$. Otherwise, $\Ao$ and $\Bo$ are {\em incompatible}. We further denote by $\cOO$ the convex compact set of all pairs of compatible observables with outcomes in $\Omega$, i.e.,
\equano{
\cOO = \{(\Co_{[1]},\Co_{[2]})\mid \Co\in\OO{\Omega^2}\} \,.
}
The relative boundary and the set of all extreme points of $\cOO$ are indicated by $\partial\cOO$ and $\ext{\cOO}$, respectively \cite{CA70}.

Now, let $\Omega$ be any set with cardinality $\abs{\Omega} = d$, and suppose the vectors $\{\phii_x\mid x\in\Omega\}$ and $\{\psi_y\mid y\in\Omega\}$ are two mutually unbiased bases (MUB) of $\hh$. That is, they are orthonormal bases of $\hh$ satisfying the mutal unbiasedness condition
\equasi{eq:MUB_condition}{
\abs{\ip{\phii_x}{\psi_y}} = \frac{1}{\sqrt{d}} \qquad \forall x,y\in\Omega \,.
}
We are interested in the following two sharp observables $\Ao,\Bo\in\OO{\Omega}$
\equasi{eq:MUB}{
\Ao(x) = \kb{\phii_x}{\phii_x}\,,\qquad\qquad\Bo(y) = \kb{\psi_y}{\psi_y}
}
and in their smeared versions
\equasi{eq:noisy-MUB}{
\Ao_\lam = \lam\Ao + (1-\lam)\Uo_\Omega\,,\qquad\qquad \Bo_\mu = \mu\Bo + (1-\mu)\Uo_\Omega
}
with unsharpness parameters $\lam,\mu\in\left[1/(1-d)\,,\,1\right]$. Note that the latter interval constitues all values of $\lam$ and $\mu$ for which \eqref{eq:noisy-MUB} actually defines two observables. In particular, whenever $\lam,\mu > 0$, the observables $\Ao_\lam$, $\Bo_\mu$ can be understood as uniformly noisy versions of the sharp observables $\Ao$, $\Bo$. For $\lam,\mu$ taking the negative values $\left[1/(1-d)\,,\,0\right)$, however, this interpretation does no longer apply. Nevertheless, with some abuse of terminology we will refer to the observables in \eqref{eq:noisy-MUB} as the {\em uniformly noisy versions} of the given two MUB for all (even negative) $\lam,\mu\in\left[1/(1-d)\,,\,1\right]$.

The problem of determining all values of $\lam$ and $\mu$ such that $\Ao_\lam$ and $\Bo_\mu$ constitute a pair of compatible observables has been addressed in \cite{CaHeToPRL19,Busch86}; the resulting compatibility region
\equano{
C_d = \bigg\{(\lam,\mu)\in \bigg[\frac{1}{1-d}\,,\,1\bigg]^2 \mid (\Ao_\lam,\Bo_\mu)\in\cOO\bigg\}
}
is depicted in Fig.~\ref{fig:MUBregion} for two different values of $d$. The extreme points $\ext{C_d}$ of the convex set $C_d$ turn out to be
\begin{enumerate}[(a)]
\item when $d=2$,\label{it:a_compatibility}
\equasi{eq:circle}{
\ext{C_2} = \{(\lam,\mu)\in\Rb^2\mid\lam^2+\mu^2 = 1\}
}
(see \cite[Cor.~4.6]{Busch86});
\item when $d\geq 3$,\label{it:b_compatibility}
\equano{
\ext{C_d} = \Gamma_d \cup \{\bgam_d\} \,,
}
where $\Gamma_d$ is the part of an ellipse described by the equations
\equasi{eq:ellipse}{
\Gamma_d : \begin{cases}
(d-1)(\lam + \mu) \geq d-3 \\
d(\lam^2 + \mu^2) + 2(d-2) \lam\mu - 2(d-2) (\lam + \mu) = 4-d
\end{cases}
}
and $\bgam_d$ is the point
\equasi{eq:point}{
\bgam_d = \bigg(\frac{1}{1-d}\,,\,\frac{1}{1-d}\bigg)
}
(see \cite[Thm.~S3 of the Supplementary Material]{CaHeToPRL19}).
\end{enumerate}

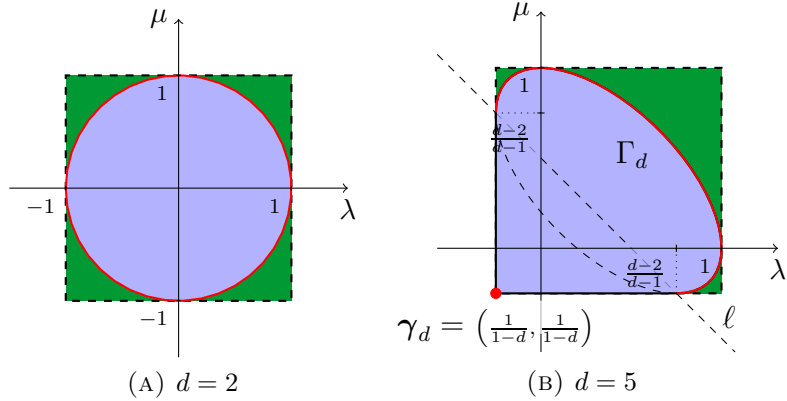
\begin{figure}[h!]\def\xy{1.5}
\centering
\subfloat[$d=2$]{\label{fig:MUBregion_a}
\begin{tikzpicture}[xscale=1.50,yscale=1.50,
declare function={
rsin(\x) = sin(deg(\x));
rcos(\x) = cos(deg(\x));
}]

\draw [thick, dashed, fill = darkgreen]
(-1,-1) -- (-1,1) -- plot (1,1) -- (1,-1) -- (-1,-1) ;

\draw [color=blue!30!white, domain=-1:1, samples=40, fill = blue!30!white]
plot ({rcos(\x*pi)}, {rsin(\x*pi)});

\draw [thick, red, domain=-1.0:1.0, samples=40]
plot ({rcos(\x*pi)}, {rsin(\x*pi)});

\draw ({rcos(0.25*pi)}, {rsin(0.25*pi)});

\draw[->] ({-\xy},0)--(\xy,0)node[anchor=north]{\small $\lam$};
\draw[->] (0,{-\xy})--(0,\xy)node[anchor=east]{\small $\mu$};

\draw (-1,0)node[anchor=north east]{\tiny $-1$};
\draw[thin,-] (-1,-0.02)--(-1,0.02);

\draw (1,0)node[anchor=north east]{\tiny $1$};
\draw[thin,-] (1,-0.02)--(1,0.02);

\draw (0,-1)node[anchor=north east]{\tiny $-1$};
\draw[thin,-] (-0.02,-1)--(0.02,-1);

\draw (0,1)node[anchor=north east]{\tiny $1$};
\draw[thin,-] (-0.02,1)--(0.02,1);

\end{tikzpicture}}
\subfloat[$d=5$]{\label{fig:MUBregion_b}
\def\d{5}
\begin{tikzpicture}[xscale=2.4,yscale=2.4,
declare function={
ratan(\x) = rad(atan(\x));
rsin(\x) = sin(deg(\x));
rcos(\x) = cos(deg(\x));
}]

\def\th0{(pi-ratan(sqrt(\d-1)))}
\def\xymen{(-(0.46/1.2)*\xy)}
\def\xypiu{((1.05/1.2)*\xy)}

\draw [thick, dashed, fill = darkgreen]
({1/(1-\d)},{1/(1-\d)}) -- ({1/(1-\d)},1) -- (1,1) -- (1,{1/(1-\d)}) -- ({1/(1-\d)},{1/(1-\d)}) ;

\draw [thick, domain=-1:1, samples=40, fill=blue!30!white]
plot({1/(1-\d)},{1/(1-\d) + 1/2 * (\x+1)})
-- plot ({(\d*rcos(\x*\th0 + \th0) + 2 - \d)/(2*(1 - \d))}, {(\d*rcos(\x*\th0 - \th0) + 2 - \d)/(2*(1 - \d))})
-- plot({1/(1-\d) + 1/2 * (-\x+1)},{1/(1-\d)}) ;

\draw [thick,red, domain=-1:1, samples=40]
plot ({(\d*rcos(\x*\th0 + \th0) + 2 - \d)/(2*(1 - \d))}, {(\d*rcos(\x*\th0 - \th0) + 2 - \d)/(2*(1 - \d))}) ;

\draw [dashed, domain=1:{(2*pi)/\th0-1}, samples=40]
plot ({(\d*rcos(\x*\th0 + \th0) + 2 - \d)/(2*(1 - \d))}, {(\d*rcos(\x*\th0 - \th0) + 2 - \d)/(2*(1 - \d))}) ;

\draw[->] ({\xymen},0)--({\xypiu},0)node[anchor=north]{\small $\lam$};
\draw (0,{\xymen})--(0,{\xymen + 0.1});
\draw [opacity=0.3] (0,{\xymen + 0.1})--(0,{1/(1-\d)-0.04});
\draw[->] (0,{1/(1-\d)-0.04})--(0,{\xypiu})node[anchor=east]{\small $\mu$};

\draw [dashed] ({\xymen},{-\xymen+(\d-3)/(\d-1)}) -- ({-1/(\d-1)},{(\d-2)/(\d-1)}) ;
\draw [dashed,opacity=0.3] ({-1/(\d-1)},{(\d-2)/(\d-1)}) -- (0,{(\d-3)/(\d-1)});
\draw [dashed] (0,{(\d-3)/(\d-1)}) -- ({(\d-3)/(\d-1)},0) ;
\draw [dashed,opacity=0.3] ({(\d-3)/(\d-1)},0) -- ({(\d-2)/(\d-1)},{-1/(\d-1)});
\draw [dashed] ({(\d-2)/(\d-1)},{-1/(\d-1)}) -- ({-\xymen+(\d-3)/(\d-1)},{\xymen})
node[pos=0.90,above=0.1cm]{$\ell$} ;


\draw ({(\d-2)/(\d-1)},0)node[anchor=north east]{\tiny $\frac{d-2}{d-1}$};
\draw[thin,-] ({(\d-2)/(\d-1)},-0.02)--({(\d-2)/(\d-1)},0.02);
\draw[dotted,-] ({(\d-2)/(\d-1)},0)--({(\d-2)/(\d-1)},{1/(1-\d)});

\draw (1,0)node[anchor=north east]{\tiny $1$};
\draw[thin,-] (1,-0.02)--(1,0.02);


\draw (0,{(\d-2)/(\d-1)})node[anchor=north east]{\tiny $\frac{d-2}{d-1}$};
\draw[thin,-] (-0.02,{(\d-2)/(\d-1)})--(0.02,{(\d-2)/(\d-1)});
\draw[dotted,-] (0,{(\d-2)/(\d-1)})--({1/(1-\d)},{(\d-2)/(\d-1)});

\draw (0,1)node[anchor=north east]{\tiny $1$};
\draw[thin,-] (-0.02,1)--(0.02,1);

\draw ({(\d*rcos(\th0) + 2 - \d)/(2*(1 - \d))}, {(\d*rcos(-\th0) + 2 - \d)/(2*(1 - \d))}) node[anchor=north east]{$\Gamma_d$};

\draw[red,fill=red] ({-1/(\d-1)},{-1/(\d-1)}) circle (0.8pt);
\draw ({-1/(\d-1)},{-1/(\d-1)})  node[anchor=north, yshift=-2pt]{$\bgam_d =$ \tiny $\left(\frac{1}{1-d} , \frac{1}{1-d}\right)$};
\end{tikzpicture}}
\caption{The set $C_d$ of all couples $(\lam,\mu)$ for which \eqref{eq:noisy-MUB} defines two observables $\Ao_\lam,\Bo_\mu$ (green square), and the one for which these observables are compatible (blue region) for different values of the dimension $d$. The extreme points $\ext{C_d}$ (red points) constitute all values of $(\lam,\mu)$ for which the pair of observables $(\Ao_\lam,\Bo_\mu)$ may be extremal in the set $\cOO$.\label{fig:MUBregion}}
\end{figure}

The values of $\lam,\mu$ for which the pairs of observables $(\Ao_\lam,\Bo_\mu)$ are extremal in the set $\cOO$ then need to be seeked among those listed in \eqref{eq:circle}, \eqref{eq:ellipse} and \eqref{eq:point}. This is indeed the task of the present paper.

\section{Extremality in two-dimensional systems}\label{sec:warm-up}

As a simple warm-up, we consider the two-dimensional case $d=2$. Then, for a suitable choice of unit vectors $\va,\vb\in\Rb^3$, the two noisy observables \eqref{eq:noisy-MUB} can be rewritten as
$$
\Ao_\lam (x) = \frac{1}{2}\left(\id + x\lam\,\va\cdot\vsigma\right)\,,\qquad\qquad\Bo_\mu (y) = \frac{1}{2}\left(\id + y\mu\,\vb\cdot\vsigma\right) \,,
$$
where $\vsigma = (\sigma_1,\sigma_2,\sigma_3)$ are the three Pauli matrices on $\hh\simeq\Cb^2$, and we used the labeling $\Omega = \{+,-\}$ for the observable outcomes. The mutual unbiasedness condition \eqref{eq:MUB_condition} is equivalent to $\va\cdot\vb = \vnull$. By \cite[Cor.~4.6]{Busch86}, it then follows that $(\Ao_\lam,\Bo_\mu)\in\cOO$ if and only if $(\lam,\mu)$ belongs to the unit disk $C_2 = \{(\lam,\mu)\in\Rb^2 \mid \lam^2+\mu^2 \leq 1\}$.

Whenever $(\Ao_\lam,\Bo_\mu)\in\cOO$, we can thus define the four new observables
$$
\Ao^\pm (x) = \frac{1}{2}\left[\id + x(\lam\va \pm \mu\vb)\cdot\vsigma\right]\,,\qquad\qquad\Bo^\pm (y) = \frac{1}{2}\left[\id + y(\mu\vb \pm \lam\va)\cdot\vsigma\right] \,.
$$
Note that $(\Ao_\lam,\Bo_\mu) = (1/2) [(\Ao^+,\Bo^+) + (\Ao^-,\Bo^-)]$, and clearly $(\Ao^+,\Bo^+) \neq (\Ao^-,\Bo^-)$ unless $\lam = \mu = 0$. Moreover, since
$$
\frac{1}{2}\no{(\lam\va \pm \mu\vb) + (\mu\vb \pm \lam\va)} + \frac{1}{2}\no{(\lam\va \pm \mu\vb) - (\mu\vb \pm \lam\va)} = \sqrt{\lam^2+\mu^2} \,,
$$
the compatibility of $(\Ao_\lam,\Bo_\mu)$ implies that also $(\Ao^\pm,\Bo^\pm)\in\cOO$ by \cite[Thm.~4.5]{Busch86}. This makes us conclude that the pair of observables $(\Ao_\lam,\Bo_\mu)$ is not an extreme point of the set $\cOO$ for any $(\lam,\mu)\in\ext{C_2}$.

Going into more detail, an alternative reason for which $(\lam,\mu)\in\ext{C_2}$ never implies $(\Ao_\lam,\Bo_\mu)\in\ext{\cOO}$ is the following. For $(\lam,\mu)\in C_2$, a joint observable $\Co$ of $\Ao_\lam$ and $\Bo_\mu$ is given by
\equasi{eq:joint_qubit}{
\Co(x,y) = \frac{1}{4}\left(\id + x\lam\,\va\cdot\vsigma + y\mu\,\vb\cdot\vsigma\right) \,.
}
Now, the essential point is that the operators $\{\Co(x,y)\mid x,y\in\Omega\}$ are not linearly independent. The non-extremality of the margin observables $(\Ao_\lam,\Bo_\mu)$ is then a consequence of the next general result.

\begin{proposition}\label{prop:ext->indep}
Suppose $\Ao,\Bo\in\OO{\Omega}$ are any two compatible observables such that $(\Ao,\Bo)$ is an extreme point of $\cOO$, and let $\Co$ be a joint observable of $\Ao$ and $\Bo$. Then, the operators $\{\Co(x,y)\mid x,y\in\Omega\}$ are linearly independent.
\end{proposition}
\begin{proof}
Suppose the operators $\{\Co(x,y)\mid x,y\in\Omega\}$ are not linearly independent, and let $c:\Omega^2\to\Cb$ be a non-zero function such that $\sum_{x,y} c(x,y)\Co(x,y) = 0$. By possibly replacing $c$ with either $c+\overline{c}$ or $\rmi (c-\overline{c})$, it can be assumed that $c(x,y)\in\Rb$ for all $x,y$. If $\varepsilon > 0$ is small enough, we can define two observables $\Co^+,\Co^-\in\OO{\Omega^2}$ given by $\Co^\pm(x,y) = (1\pm\varepsilon c(x,y))\Co(x,y)$ for all $x,y\in\Omega$. Clearly, $\Co = (1/2)(\Co^+ + \Co^-)$ and $\Co^+\neq\Co^-$, hence $(\Ao,\Bo)$ is not an extreme point of $\cOO$ by \cite[Cor.~5]{GuCu18}.
\end{proof}

We end this section by noticing that, for $(\lam,\mu)\in\ext{C_2}$ and either $\lam \geq 0$ or $\mu \geq 0$, the joint observable \eqref{eq:joint_qubit} is of the L\"uders form
\equasi{eq:Luders1}{
\Co(x,y) = \begin{cases}
\Ao_\lam(x)^\frac{1}{2}\Bo(y)\Ao_\lam(x)^\frac{1}{2} & \text{ if $\mu\geq 0$} \\
\Bo_\mu(y)^\frac{1}{2}\Ao(x)\Bo_\mu(y)^\frac{1}{2} & \text{ if $\lam\geq 0$}
\end{cases} \,,
}
where
$$
\Ao_\lam(x)^\frac{1}{2} = \frac{1}{2\sqrt{2}}\left\{ \left[(1+\lam)^\frac{1}{2} + (1-\lam)^\frac{1}{2}\right]\id + x \left[(1+\lam)^\frac{1}{2} - (1-\lam)^\frac{1}{2}\right]\va\cdot\vsigma \right\}
$$
is the square root of the operator $\Ao_\lam(x)$ and a similar formula holds for $\Bo_\mu(y)^\frac{1}{2}$.
Remarkably, when $d\geq 3$ and $(\lam,\mu)\in\ext{C_d}$ lies in the part of an ellipse \eqref{eq:ellipse}, the L\"uders observables \eqref{eq:Luders1} still constitute joint observables of $\Ao_\lam$ and $\Bo_\mu$, as it is explained in the next section.

\section{Extremality in higher dimensions}\label{sec:higher-d}

In dimensions $d\geq 3$, we have
$$
\Ao_\lam(x)^\frac{1}{2} = u_\lam\Ao(x) + \frac{v_\lam}{\sqrt{d}}\,\id\,,\qquad\qquad \Bo_\mu(y)^\frac{1}{2} = u_\mu\Bo(y) + \frac{v_\mu}{\sqrt{d}}\,\id 
$$
for all $\lam,\mu\in\left[1/(1-d)\,,\,1\right]$, where we set
\equasi{eq:uv}{
u_\nu = \frac{1}{\sqrt{d}}\left\{[1+(d-1)\nu]^\frac{1}{2} - (1-\nu)^\frac{1}{2} \right\} \,,\qquad\qquad v_\nu = (1-\nu)^\frac{1}{2} \,.
}
Note that the real constants $u_\nu$ and $v_\nu$ satisfy
\equasi{eq:propr_uv}{
u^2_\nu + v^2_\nu + \frac{2u_\nu v_\nu}{\sqrt{d}} = 1 \ \ \text{and} \ \ v_\nu \geq 0 \quad \text{for all $\nu$}\,,\qquad u_\nu \geq 0 \ \ \text{iff} \ \ \nu\geq 0 \,.
}
Since $\Ao(x)\Bo(y)\Ao(x) = (1/d)\Ao(x)$ and $\Bo(y)\Ao(x)\Bo(y) = (1/d)\Bo(y)$ by the mutual unbiasedness condition \eqref{eq:MUB_condition}, the L\"uders observables \eqref{eq:Luders1} become
\equasi{eq:Luders2}{
\begin{aligned}
\Ao_\lam(x)^\frac{1}{2}\Bo(y)\Ao_\lam(x)^\frac{1}{2} = \frac{1}{d}\left[u^2_\lam\Ao(x) + v_\lam^2\Bo(y)\right] + \frac{u_\lam v_\lam}{\sqrt{d}}\left[\Ao(x)\Bo(y) + \Bo(y)\Ao(x)\right] , \\
\Bo_\mu(y)^\frac{1}{2}\Ao(x)\Bo_\mu(y)^\frac{1}{2} = \frac{1}{d}\left[v_\mu^2\Ao(x) + u^2_\mu\Bo(y)\right] + \frac{u_\mu v_\mu}{\sqrt{d}}\left[\Ao(x)\Bo(y) + \Bo(y)\Ao(x)\right] .
\end{aligned}
}
Evaluating their margin observables yields the pairs $(\Ao_\lam,\Bo_{\gamma_\lam})$ and $(\Ao_{\gamma_\mu},\Bo_\mu)$, respectively, where
\begin{subequations}\label{eq:gammanu}
\begin{align}
\gamma_\nu & = v^2_\nu + \frac{2u_\nu v_\nu}{\sqrt{d}} = 1-u^2_\nu \label{eq:gammanu1} \\
& = \frac{1}{d} \left\{(d-2)(1-\nu) + 2(1-\nu)^\frac{1}{2}\left[1+(d-1)\nu\right]^\frac{1}{2}\right\} \,. \label{eq:gammanu2}
\end{align}
\end{subequations}
From the expression \eqref{eq:gammanu2}, we see that $\gamma_\nu \geq 0$ for all $\nu\in\left[1/(1-d)\,,\,1\right]$. Moreover, by direct inspection, both points $(\lam,\mu) = (\lam,\gamma_\lam)$ and $(\lam,\mu) = (\gamma_\mu,\mu)$ lie on the curve $\Gamma_d$ defined in \eqref{eq:ellipse}. Thus, the mappings $\lam\mapsto(\lam,\gamma_\lam)$ and $\mu\mapsto(\gamma_\mu,\mu)$ are two different parametrizations of $\Gamma_d$: they parametrize the two parts of $\Gamma_d$ with either $\mu\geq 0$ or $\lam\geq 0$, respectively. These parametrizations overlap when both $\lam\geq 0$ and $\mu\geq 0$; in this case, since $u_\lam\geq 0$ and $u_\mu\geq 0$, by solving \eqref{eq:gammanu1} with respect to $u_\nu$ and comparing the result with \eqref{eq:uv}, we obtain that $u_\mu = v_\lam$ and $u_\lam = v_\mu$. We then conclude that the L\"uders observables \eqref{eq:Luders2} coincide for $(\lam,\mu)\in\Gamma_d\cap\Rb^2_+$, where $\Rb_+ = [0,+\infty)$. Hence, for all $(\lam,\mu)\in\Gamma_d$, we can define a joint observable $\Co$ of $\Ao_\lam$ and $\Bo_\mu$ as in \eqref{eq:Luders1} without any ambiguity.

Having constructed a joint observable $\Co$, we are now in position to state the first main result of the paper.

\begin{theorem}\label{thm:extreme_Gamma_d}
Suppose $(\lam,\mu)\in\Gamma_d$, where the curve $\Gamma_d$ is defined in \eqref{eq:ellipse}. Then, the pair of observables $(\Ao_\lam,\Bo_\mu)$ is an extreme point of the convex set $\cOO$ if and only if both the following conditions are satisfied:
\begin{enumerate}[(i)]
\item $\lam\neq 0$ and $\mu\neq 0$;\label{it:a_thm_extreme_Gamma_d}
\item the $d^2\times d^2$ real symmetric matrix $\Lam$ with entries\label{it:b_thm_extreme_Gamma_d}
\equasi{eq:Haagerup_matrix}{
\Lambda_{(x,y),(z,t)} = d\,\Re\left[\ip{\psi_y}{\phii_z}\ip{\phii_z}{\psi_t}\ip{\psi_t}{\phii_x}\ip{\phii_x}{\psi_y}\right]
}
does not have $-1$ among its eigenvalues. (Here we regard $\Lambda$ as a square matrix in which rows are indexed by elements $(x,y)\in\Omega^2$ and columns by elements $(z,t)\in\Omega^2$.)
\end{enumerate}
\end{theorem}

The matrix $\Lambda$ -- or, more precisely, its complex version defined without taking the real part in \eqref{eq:Haagerup_matrix} -- was introduced for the first time by Haagerup in order to determine all orthogonal maximal abelian $\ast$-subalgebras of the complex $d\times d$ matrices (see \cite[Remark 2.10]{Haagerup96}). This problem is related to the classification of all complex $d\times d$ Hadamard matrices (see \cite[Eq.~(11) and Lemma 2.1]{TaZy06}) and thus of all pairs of MUB in dimension $d$. Note that $\Lam$ is unaltered by rephasing the two MUB $\{\phii_x\mid x\in\Omega\}$ and $\{\psi_y\mid y\in\Omega\}$, while relabeling or even exchanging them only permutes its row and column indices. In particular, the spectrum ${\rm sp}(\Lam)$ solely depends on the {\em equivalence class} of the (unordered) pair of MUB at hand \cite[Sec.~5.1]{DuEnBeZy10}.

In Sec.~\ref{subsec:Fourier} below, we will provide a whole family of MUB for which ${\rm sp}(\Lam)$ can be analytically determined; this family is made up of all pairs of MUB that are conjugate by the Fourier transform of a finite abelian group. In the subsequent Sec.~\ref{subsec:numeric}, by means of computer-assisted calculations, we will test the condition $-1\in {\rm sp}(\Lam)$ for some other examples given in \cite{TaZy06}.

The proof of Thm.~\ref{thm:extreme_Gamma_d} relies on the next partial converse of Prop.~\ref{prop:ext->indep}, along with the subsequent more specific technical lemma.

\begin{proposition}\label{prop:indep+rk1->ext}
Suppose $\Ao,\Bo\in\OO{\Omega}$ are any two compatible observables, and let $\Co$ be a joint observable of $\Ao$ and $\Bo$. If the operators $\{\Co(x,y)\mid x,y\in\Omega\}$ are linearly independent and $\rank{\Co(x,y)} = 1$ for all $x,y\in\Omega$, then $(\Ao,\Bo)$ is an extreme point of $\cOO$.
\end{proposition}
\begin{proof}
We show that the observable $\Co$ is an extreme point of the convex set $\OO{\Omega^2}$, and then its marginals $(\Ao,\Bo)$ are extremal in $\cOO$ by \cite[Cor.~5]{GuCu18}. Indeed, suppose that $\Co = \lam_+\Co^+ + \lam_-\Co^-$ for some $\Co^+,\Co^-\in\OO{\Omega^2}$ and $\lam_+,\lam_- \in (0,1)$ with $\lam_+ + \lam_- = 1$. It follows that $0\leq \Co^\pm(x,y)\leq (1/\lam_\pm)\Co(x,y)$, hence by the rank-$1$ hypothesis there exist functions $c^\pm:\Omega^2\to\Rb_+$ such that $\Co^\pm(x,y) = c^\pm(x,y)\Co(x,y)$ for all $x,y$. Combining the normalization condition
$$
\sum_{x,y\in\Omega} c^\pm(x,y)\Co(x,y) = \sum_{x,y\in\Omega} \Co^\pm(x,y) = \id = \sum_{x,y\in\Omega} \Co(x,y)
$$
with the assumed linear independence of the set $\{\Co(x,y)\mid x,y\in\Omega\}$, it follows that $c^+ = c^- = 1$. Thus, $\Co^+ = \Co^- = \Co$, which proves the claimed extremality of $\Co$.
\end{proof}

\begin{lemma}\label{lem:K+E}
Suppose $\Ao,\Bo\in\OO{\Omega}$ are the two sharp observables defined in \eqref{eq:MUB}, and, for constant numbers $a,b,c,e\in\Rb$, let
$$
C(x,y) = a\Ao(x) + b\Bo(y) + c[\Ao(x)\Bo(y) + \Bo(y)\Ao(x)] + e\id \qquad \forall x,y\in\Omega\,.
$$
Then, the operators $\{C(x,y)\mid x,y\in\Omega\}$ are linearly independent if and only if both condition \eqref{it:b_thm_extreme_Gamma_d} of Thm.~\ref{thm:extreme_Gamma_d} hold and all the following inequalities are satisfied:
\equasi{eq:Kinvert}{
c \neq 0\,, \qquad da+2c\neq 0\,, \qquad db+2c\neq 0\,, \qquad da+db+2c+d^2 e\neq 0\,.
}
\end{lemma}
\begin{proof}
We rewrite
$$
C(x,y) = \sum_{z,t\in\Omega} K_{(x,y),(z,t)} E(z,t) \,,
$$
where
\begin{align*}
E(z,t) & = \frac{1}{2}\left[ \Ao(z)\Bo(t) + \Bo(t)\Ao(z) \right] \,, \\
K_{(x,y),(z,t)} & = a\delta_{x,z} + b\delta_{y,t} + 2c\delta_{x,z}\delta_{y,t} + e
\end{align*}
and $\delta_{x,z}$, $\delta_{y,t}$ are the usual Kronecker deltas.
The linear independence of the $d^2$ operators $\{C(x,y)\mid x,y\in\Omega\}$ then amounts to the linear independence of the $d^2$ operators $\{E(x,y)\mid x,y\in\Omega\}$ and the invertibility of the $d^2\times d^2$ transition matrix $K$ with entries $K_{(x,y),(z,t)}$. The first of the latter two conditions is equivalent to condition \eqref{it:b_thm_extreme_Gamma_d} of Thm.~\ref{thm:extreme_Gamma_d}, since
$$
\tr{E(z,t)^*E(x,y)} = \frac{1}{2d}\left( \delta_{x,z}\delta_{y,t} + \Lambda_{(x,y),(z,t)} \right) \,.
$$
In order to unravel the second condition, we simplify the expression for $K$ by introducing the $d\times d$ rank-$1$ orthogonal projection matrix $P$ with $P_{x,z} = 1/d \ \ \forall x,y \in\Omega$ and its orthogonal complement $P^\perp = \id - P$. This gives the following spectral decomposition of $K$:
$$
K = 2c\, P^\perp\otimes P^\perp + (da+2c)\, P^\perp\otimes P + (db+2c)\, P\otimes P^\perp + (da+db+2c+d^2 e)\, P\otimes P\,.
$$
From it, we see that $K$ is invertible if and only if all inequalities \eqref{eq:Kinvert} are satisfied. This concludes the proof of the lemma.
\end{proof}

\begin{proof}[of Thm.~\ref{thm:extreme_Gamma_d}]
We only prove the case $\mu\geq 0$, the case $\lam\geq 0$ being symmetric. If $\mu\geq 0$, we have seen that a joint observable of $\Ao_\lam$ and $\Bo_\mu$ is the first of the two L\"uders observables defined in \eqref{eq:Luders2}, that is, the rank-$1$ observable $\Co$ given by
$$
\Co(x,y) = \frac{1}{d}\left[u^2_\lam\Ao(x) + v_\lam^2\Bo(y)\right] + \frac{u_\lam v_\lam}{\sqrt{d}}\left[\Ao(x)\Bo(y) + \Bo(y)\Ao(x)\right] \,.
$$
Moreover, $\mu = \gamma_\nu$ with $\gamma_\nu$ given by \eqref{eq:gammanu} in this case.
By Props.~\ref{prop:ext->indep} and \ref{prop:indep+rk1->ext}, $(\Ao_\lam,\Bo_\mu)\in\ext{\cOO}$ if and only if the operators $\{\Co(x,y)\mid x,y\in\Omega\}$ are linearly independent. By Lemma \ref{lem:K+E}, this is equivalent to condition \eqref{it:b_thm_extreme_Gamma_d} of Thm.~\ref{thm:extreme_Gamma_d} along with inequalities \eqref{eq:Kinvert} for the constants $a = u^2_\lam / d$, $b = v^2_\lam / d$, $c = u_\lam v_\lam / \sqrt{d}$ and $e=0$. According to \eqref{eq:propr_uv}, we have $da+2c = 1-v^2_\lam$, $db+2c = 1-u^2_\lam$ and $da+db+2c+d^2 e = 1$, hence inequalities \eqref{eq:Kinvert} hold if and only if $u_\lam \neq 0$, $v_\lam \neq 0$, $u^2_\lam \neq 1$ and $v^2_\lam \neq 1$. By \eqref{eq:uv} and \eqref{eq:gammanu1}, we have $u_\lam = \pm v_\mu$ (with `$+$' if $\lam\geq 0$ and `$-$' if $\lam < 0$ by \eqref{eq:propr_uv}). Therefore, inequalities \eqref{eq:Kinvert} are equivalent to condition \eqref{it:a_thm_extreme_Gamma_d} of Thm.~\ref{thm:extreme_Gamma_d}. This concludes the proof of the theorem.
\end{proof}

\section{Examples}\label{sec:examples}

\subsection{Two Fourier conjugate MUB}\label{subsec:Fourier}
As a first example, we apply Thm.~\ref{thm:extreme_Gamma_d} to the case of two Fourier conjugate MUB; that is, we assume that $\Omega$ is an order $d$ abelian group and
\equasi{eq:MUB_Fourier}{
\ip{\phii_x}{\psi_y} = \frac{1}{\sqrt{d}} \pair{x}{y} ,
}
where $\pair{\cdot}{\cdot}$ is some non-degenerate symmetric bicharacter of $\Omega$. Here, we recall that a {\em non-degenerate symmetric bicharacter} of an abelian group $\Omega$ is any map $\pair{\cdot}{\cdot} : \Omega\times\Omega\to\{z\in\Cb\mid\abs{z} = 1\}$ such that
\begin{enumerate}[(i)]
\item $\pair{x}{y} = \pair{y}{x}$ for all $x,y\in\Omega$;
\item $\pair{x_1+x_2}{y} = \pair{x_1}{y}\pair{x_2}{y}$ for all $x_1,x_2,y\in\Omega$, where the composition law of $\Omega$ is written additively;
\item denoting by $\widehat{\Omega}$ the dual group of $\Omega$ \cite[Ch.~I, \S 9]{LangAlgebra}, the mapping $x\mapsto\pair{x}{\cdot}$ is a group isomorphism of $\Omega$ onto $\widehat{\Omega}$.
\end{enumerate}

The case with $\Omega$ being the cyclic group $\Zb_d = \{0,1,\ldots,d-1\}$ endowed with the bicharacter $\pair{x}{y} = \rme^{2\pi\rmi xy/d}$ was already treated in \cite{CaHeTo12}. In particular, \cite[Prop.~5]{CaHeTo12} states that in this case the observables $\Ao_\lam$ and $\Bo_\mu$ have a unique joint observable $\Co$ whenever the point $(\lam,\mu)$ lies in $\Gamma_d\cap\Rb^2_+$, and for such observable $\rank{\Co(x,y)} = 1$ for all $x,y$. Further, by \cite[Prop.~9]{CaHeTo12} the operators $\{\Co(x,y) \mid x,y\in\Omega\}$ are linearly independent if and only if $d$ is an odd number and $(\lam,\mu)\in\Gamma_d\cap\Rb^2_{+\ast}$, where $\Rb_{+\ast} = (0,+\infty)$. Combining these earlier results with Props.~\ref{prop:ext->indep} and \ref{prop:indep+rk1->ext} of the present paper, we conclude that in the cyclic group case a point $(\lam,\mu)\in\Gamma_d\cap\Rb^2_{+\ast}$ yelds a pair $(\Ao_\lam,\Bo_\mu)\in\ext{\cOO}$ if and only if $d$ is odd. The next theorem extends this result to arbitrary abelian groups $\Omega$ and points $(\lam,\mu)\in\Gamma_d$ with $\lam\neq 0$ and $\mu\neq 0$.

\begin{theorem}\label{thm:Fourier}
Suppose $\Omega$ is an order $d$ abelian group endowed with the non-degenerate symmetric bicharacter $\pair{\cdot}{\cdot}$, and assume the two MUB $\{\phii_x\mid x\in\Omega\}$ and $\{\psi_y\mid y\in\Omega\}$ satisfy \eqref{eq:MUB_Fourier}. Further, let $(\lam,\mu)\in\Gamma_d$ with $\lam\neq 0$ and $\mu\neq 0$. Then, the pair of observables $(\Ao_\lam,\Bo_\mu)$ is an extreme point of the convex set $\cOO$ if and only if $d$ is odd.
\end{theorem}
\begin{proof}
We need to prove that under the hypotheses in the statement, condition \eqref{it:b_thm_extreme_Gamma_d} of Thm.~\ref{thm:extreme_Gamma_d} is satisfied if and only if $d$ is odd. This is done in the next proposition, which we split from the present proof for later purposes (see the proof of Thm.~\ref{thm:gamma_d} below).
\end{proof}

\begin{proposition}\label{prop:sp(Lam)-Fourier}
Suppose $\Omega$ is an order $d$ abelian group endowed with the non-de\-gen\-er\-ate symmetric bicharacter $\pair{\cdot}{\cdot}$, and assume the two MUB $\{\phii_x\mid x\in\Omega\}$ and $\{\psi_y\mid y\in\Omega\}$ satisfy \eqref{eq:MUB_Fourier}. Then, the matrix $\Lam$ defined in \eqref{eq:Haagerup_matrix} does not have $-1$ among its eigenvalues if and only if $d$ is odd.
\end{proposition}
\begin{proof}
Assuming \eqref{eq:MUB_Fourier}, the matrix \eqref{eq:Haagerup_matrix} becomes $\Lam = (1/2)(\Lam^+ + \Lam^-)$, where
$$
\Lambda^+_{(x,y),(z,t)} = \frac{1}{d} \pair{z-x}{t-y} = \Lambda^-_{(x,t),(z,y)} \,.
$$
The orthogonality relations between elements of $\widehat{\Omega}$ \cite[Ch.~XVIII, \S 5]{LangAlgebra} give
$$
\sum_{z\in\Omega}\pair{x-y}{z} = d\,\delta_{x,y} \,.
$$
As a consequence, the vectors $\{w^{r,s} \mid r,s\in\Omega\}$ with
$$
w^{r,s}_{(z,t)} = \frac{1}{d^2} \pair{r}{z}\pair{s}{t} \qquad \forall (z,t)\in\Omega^2 
$$
constitute an orthonormal basis of $\Cb^{d^2}$. Moreover,
\aligno{
(\Lam^+ w^{r,s})_{(x,y)} & = \frac{1}{d^3} \sum_{z,t\in\Omega} \pair{z-x}{t-y} \pair{r}{z}\pair{s}{t} \\
& = \frac{1}{d^3} \pair{r}{x}\pair{s}{y} \sum_{z',t'\in\Omega} \pair{z'}{t'} \pair{r}{z'}\pair{s}{t'} \\
& = \frac{1}{d^2} \pair{r}{x}\pair{s}{y} \sum_{t'\in\Omega} \delta_{-r,t'} \pair{s}{t'} \\
& = \overline{\pair{s}{r}} \, w^{r,s}_{(x,y)}\,,
}
where in the second equality we made the substitutions $z = z'+x$ and $t = t'+y$, and in the third one we used the orthogonality relations. We similarly deduce that $\Lam^- w^{r,s} = \pair{s}{r} \, w^{r,s}$. Thus, the spectrum of $\Lam$ is the set
$$
{\rm sp}(\Lam) = \left\{\Re\left[\pair{r}{s}\right] \mid r,s\in\Omega\right\}\,.
$$
As $\abs{\pair{r}{s}} = 1$ for all $r,s$, we have $-1\in{\rm sp}(\Lam)$ if and only if $\pair{r}{s} = -1$ for some $r,s$. We claim that the latter condition is equivalent to $d$ being an even number. Indeed, if $\pair{r}{s} = -1$, then $r$ and $s$ generate two cyclic subgroups of even orders in $\Omega$; since the orders of these subgroups need divide the order of $\Omega$ \cite[Ch.~I, Prop.~4.1]{LangAlgebra}, it follows that $d$ is even. Conversely, if $d$ is even, then there exists an element $r\in\Omega$ such that $2r = 0$ \cite[Ch.~I, Lemma 6.1]{LangAlgebra}. Since the mapping $x\mapsto\pair{x}{\cdot}$ is a group isomorphism of $\Omega$ onto $\widehat{\Omega}$, it must be $\pair{r}{s}\neq 1$ for some $s\in\Omega$, hence $\pair{r}{s} = -1$ as $\pair{r}{s}^2 = \pair{2r}{s} = 1$.
\end{proof}

\subsection{Two non-Fourier conjugate MUB in low dimensions}\label{subsec:numeric}

All the (unordered) pairs of MUB in dimensions $d=2$, $3$ and $5$ are equivalent to the pair that is conjugate with respect to the Fourier transform of the corresponding cyclic groups $\Zb_d$ \cite[Prop.~2.1 and Thm.~2.2]{Haagerup96}. Here, equivalence is understood in the sense of \cite[Sec.~5.1]{DuEnBeZy10}, and it amounts to the equivalence of the complex Hadamard matrices associated with the MUB at hand up to conjugate transposition (see \cite[Def.~2.2]{TaZy06} for the definition of equivalent complex Hadamard matrices). We recall that the Hadamard matrix associated with the pair of MUB $\{\phii_x\mid x\in\Omega\}$ and $\{\psi_y\mid y\in\Omega\}$ is the $d\times d$ unitary matrix $H$ defined by
\equano{
H_{x,y} = \ip{\phii_x}{\psi_y} \,,
}
so that \eqref{eq:Haagerup_matrix} rewrites
\equano{
\Lambda_{(x,y),(z,t)} = d\,\Re\left(\overline{H}_{z,y} H_{z,t} \overline{H}_{x,t} H_{x,y}\right) \,.
}

The smaller dimension $d$ in which Thm.~\ref{thm:Fourier} does not exhaust all possible MUB is thus $d=4$. In this case, there exists a continuous $1$-parameter family of inequivalent complex Hadamard matrices \cite[Prop.~2.1]{Haagerup96}, labeled by $a\in [0,\pi)$. For all these matrices, we symbolically computed the eigenvalues of $\Lam$ by means of the Wolfram Mathematica\textsuperscript{\textregistered} software, and we found that $-1\in{\rm sp}(\Lam)$ independently of the value of $a$. Therefore, no uniformly noisy versions of any two MUB yield extreme points of the convex set $\cOO$ for $d=4$.

We repeated the same computer-assisted evaluation of the eigenvalues of $\Lam$ for some complex Hadamard matrices in dimensions $d=6$ and $7$, all taken from \cite{TaZy06} and inequivalent to the Fourier conjugate pairs; see Table \ref{tab:eigenvalues} for the obtained results. We always found $-1\in{\rm sp}(\Lam)$ in dimension $d=6$, and, remarkably, we found $-1\in{\rm sp}(\Lam)$ also in one case with $d=7$. The latter result shows that in odd dimensions not all pairs of MUB can be used to construct extreme points of $\cOO$, as instead one might have expected by looking at the Fourier conjugate case. It also proves that the existence of inequivalent pairs of MUB actually results in different geometric properties of the respective uniformly noisy versions, a fact that was already observed for triplets of inequivalent MUB by comparing the respective noise robustness in \cite{DeSkFrBr19}.

Still in dimensions $d=6$ and $7$, the evaluation of the eigenvalues of $\Lam$ for the other parametric families of Hadamard matrices provided by \cite{TaZy06} is a computationally demanding task, as it is for non-Fourier conjugate MUB in higher dimensions. We leave it as an open problem whether, in contrast to the Fourier conjugate case, $-1\notin{\rm sp}(\Lam)$ for some pairs of MUB in even dimension.

\begin{table}[h]
\smallskip
\Rotatebox{-270}{
{
\Tiny
{
\begin{tabularx}{1.\textwidth}{ c | c  c   c  c }
\toprule
dimension & matrix &  parameter& eigenvalues & mult.\\
\midrule
\multirow{4}{*}{$d = 4$} &
\multirow{4}{*}{
$
\frac{1}{\sqrt{4}}
\begin{pmatrix}
 1 & 1 & 1 & 1 \\
 1 & \rmi \rme^{\rmi a} & -1 & -\rmi \rme^{\rmi a} \\
 1 & -1 & 1 & -1 \\
 1 & -\rmi \rme^{\rmi a} & -1 & \rmi \rme^{\rmi a} \\
\end{pmatrix}
$} & \multirow{4}{*}{\footnotesize{$a\in [0,\pi)$}}  & $-1$ & $4$ \\
&& & $-\left| \sin (a)\right|$  & $2$\\
&& & $\left| \sin (a)\right|$ &$2$\\
&& & $1$ & $8$ \\[1mm]
\midrule
\multirow{20}{*}{$d=6$} &
\multirow{6}{*}{
$
\frac{1}{\sqrt{6}}
\begin{pmatrix}
 1 & 1 & 1 & 1 & 1 & 1 \\
 1 & -1 & \rmi & -\rmi & -\rmi & \rmi \\
 1 & \rmi & -1 & \rmi \rme^{\rmi a} & -\rmi \rme^{\rmi a} & -\rmi \\
 1 & -\rmi & \rmi \rme^{-\rmi a} & -1 & \rmi & -\rmi \rme^{-\rmi a} \\
 1 & -\rmi & -\rmi \rme^{-\rmi a} & \rmi & -1 & \rmi \rme^{-\rmi a} \\
 1 & \rmi & -\rmi & -\rmi \rme^{\rmi a} & \rmi \rme^{\rmi a} & -1 \\
\end{pmatrix}
$} & \multirow{6}{*}{\footnotesize{$a\in [0,2\pi)$}} & &\\
&&& -1  & 5 \\
&&& -2/3 & 6 \\
&&& 0 & 10 \\
&&& 1 & 15\\
&&&&\\[2.5mm]
\cline{2-5}
&&&&\\[-1.7mm]
&
\multirow{7}{*}[-1.4mm]{
$
\frac{1}{\sqrt{6}}
\begin{pmatrix}
 1 & 1 & 1 & 1 & 1 & 1 \\
 1 & -1 & -\omega  & -\omega ^2 & \omega ^2 & \omega  \\
 1 & -\overline{\omega } & 1 & \omega ^2 & -\omega ^3 & \omega ^2 \\
 1 & -\overline{\omega }^2 & \overline{\omega }^2 & -1 & \omega ^2 & -\omega ^2 \\
 1 & \overline{\omega }^2 & -\overline{\omega }^3 & \overline{\omega }^2 & 1 & -\omega  \\
 1 & \overline{\omega } & \overline{\omega }^2 & -\overline{\omega }^2 & -\overline{\omega } & -1 \\
\end{pmatrix}
$}
&  \multirow{7}{*}[-1.4mm]{\footnotesize{$\omega = \frac{1-\sqrt{3}\pm\rmi \sqrt[4]{12} }{2} \quad$}} & $-1$ & $5$ \\
&&& $-(\sqrt{3}-1)$ & 2 \\
&&& $- \Big[\sqrt{7  (7-4 \sqrt{3} )}+ 2 \sqrt{3} -3 \Big]/2$ & 4\\
&&& $-(2-\sqrt{3} )$ & 4 \\
&&& $ \Big[\sqrt{7  (7-4 \sqrt{3} )}- 2 \sqrt{3} +3 \Big]/2$ & 4\\
&&& $3 \sqrt{3}-5$ & 2 \\
&&& 1 & 15\\
\cline{2-5}\\[0.5mm]
&
\multirow{6}{*}[1.5mm]{
$
\frac{1}{\sqrt{6}}
\begin{pmatrix}
 1 & 1 & 1 & 1 & 1 & 1 \\
 1 & 1 & \omega  & \omega  & \omega ^2 & \omega ^2 \\
 1 & \omega  & 1 & \omega ^2 & \omega ^2 & \omega  \\
 1 & \omega  & \omega ^2 & 1 & \omega  & \omega ^2 \\
 1 & \omega ^2 & \omega ^2 & \omega  & 1 & \omega  \\
 1 & \omega ^2 & \omega  & \omega ^2 & \omega  & 1 \\
\end{pmatrix}
$}
&  \multirow{6}{*}[1mm]{\footnotesize{$\omega = \rme^{\frac{2\rmi \pi }{3}}$}} &&   \\
&&& $-1$ & $9$ \\
&&& $1/4$ & $16$ \\
&&& $1$ & $11$\\
&&&&\\
&&&&\\
\midrule
&&&&\\[-2.7mm]
\multirow{21}{*}{$d=7$} &
\multirow{15}{*}{
$
\frac{1}{\sqrt{7}}
\begin{pmatrix}
 1 & 1 & 1 & 1 & 1 & 1 & 1 \\
 1 & \omega  & \omega ^4 & \omega ^5 & \omega ^3 & \omega ^3 & \omega  \\
 1 & \omega ^4 & \omega  & \omega ^3 & \omega ^5 & \omega ^3 & \omega  \\
 1 & \omega ^5 & \omega ^3 & \omega  & \omega ^4 & \omega  & \omega ^3 \\
 1 & \omega ^3 & \omega ^5 & \omega ^4 & \omega  & \omega  & \omega ^3 \\
 1 & \omega ^3 & \omega ^3 & \omega  & \omega  & \omega ^4 & \omega ^5 \\
 1 & \omega  & \omega  & \omega ^3 & \omega ^3 & \omega ^5 & \omega ^4 \\
\end{pmatrix}
$}
&  \multirow{15}{*}{\footnotesize{$ \omega = \rme^{\frac{\rmi \pi }{3}}$}} & -1 & 1   \\
&&& $ -(9+\sqrt{22})/14$ & $2$ \\
&&& ${-(3\sqrt{65}+1)/28}$ & $1$\\
&&& $-11/14$ & $1$\\
&&& $-(5+3\sqrt{2})/14$ & $5$ \\
&&& $-1/2$ & $1$ \\
&&& $ -(\sqrt{22}-9)/14$ & $2$\\
&&& $-(5-{3}\sqrt{2})/14$ & $5$\\
&&& ${(3\sqrt{65}-1)/28}$ & $1$\\
&&& ${13/14}$ & $2$\\
&&& $1$ & $16$ \\
&&& $\ast$ & $3$\\
&&& $\dag$ & $3$\\
&&& $\flat$ & $3$\\
&&& $\sharp$ & $3$\\\\[-1.5mm]
\cline{2-5}\\[0.5mm]
 &
\multirow{6}{*}[1.5mm]{
$
\frac{1}{\sqrt{7}}
\begin{pmatrix}
 1 & 1 & 1 & 1 & 1 & 1 & 1 \\
 1 & \omega  & 1 & \overline{\omega } & \omega  & \overline{\omega } & 1 \\
 1 & \omega  & \omega  & \overline{\omega } & 1 & 1 & \overline{\omega } \\
 1 & \omega ^2 & \omega ^2 & \omega  & \omega  & 1 & \omega  \\
 1 & 1 & \omega  & 1 & \omega  & \overline{\omega } & \overline{\omega } \\
 1 & \omega ^2 & \omega  & \omega  & \omega ^2 & \omega  & 1 \\
 1 & \omega  & \omega ^2 & 1 & \omega ^2 & \omega  & \omega  \\
\end{pmatrix}
$}
&  \multirow{6}{*}{\footnotesize{$ \omega = \frac{-3\pm\rmi \sqrt{7}}{4} \quad $}}& $-\sqrt{57}/8$ & $8$ \\
&&& $-3/4$ & $8$\\
&&&$-\sqrt{2}/4$& $6$\\
&&&$\sqrt{2}/4$& $6$\\
&&& $\sqrt{57}/8$ & $8$ \\
&&& $1$ & $13$ \\[2mm]
\bottomrule
\end{tabularx}}}}

\smallskip
\caption{The table presents those Hadamard matrices for which the spectrum of the matrix $\Lambda$ defined by \eqref{eq:Haagerup_matrix} has been evaluated by using Mathematica\textsuperscript{\textregistered}.
The symbols $\ast,\dag,\flat,\sharp$ denote the roots of the polynomial $19208 \lam^4 + 15092 \lam^3 - 12642 \lam^2 - 6167 \lam + 3031$, and ${\rm mult.} = {\rm multiplicity}$.\label{tab:eigenvalues}}
\end{table}

\clearpage
\newpage

\section{The exceptional compatible pair in dimension $d\geq 3$}\label{sec:exceptional}

According to the discussion in Sec.~\ref{sec:prel}, there is still one value of the parameters $(\lam,\mu)$ for which the pair of observables $(\Ao_\lam,\Bo_\mu)$ may be an extreme point of the set $\cOO$, namely, in dimension $d\geq 3$, the lower left vertex $\bgam_d = (1/(1-d)\,,\,1/(1-d))$ of the compatibility region depicted in Fig.~\ref{fig:MUBregion_b}. Actually, the next result shows that only in one case this point gives rise to extremal compatible observables.

\begin{theorem}\label{thm:gamma_d}
For $d\geq 3$, the pair of observables $\big(\Ao_{1/(1-d)}\,,\,\Bo_{1/(1-d)}\big)$ is an extreme point of the convex set $\cOO$ if and only if $d=3$.
\end{theorem}

For $d\geq 3$, a joint observable of $\Ao_{1/(1-d)}$ and $\Bo_{1/(1-d)}$ is given by
\equasi{eq:C_gamma_d}{
\Co(x,y) = \frac{1}{d(d-2)}\left\{ \id -\frac{d}{d-1} \left[ \Ao(x) + \Bo(y) - \left(\Ao(x)\Bo(y) + \Bo(y)\Ao(x)\right)\right]\right\}
}
(see \cite[Prop.~S10 of the Supplementary Material]{CaHeToPRL19}). Indeed, it follows by direct inspection that $\Ao_{1/(1-d)}$ and $\Bo_{1/(1-d)}$ are the margins of $\Co$ given in \eqref{eq:C_gamma_d}; in particular, $\Co$ is normalized. However, since $\Co$ is not of the L\"uders form \eqref{eq:Luders2}, we still need to show its positivity. To this aim, observe that, for any $x,y\in\Omega$, the operator
$$
\Pi(x,y) = \frac{d}{d-1} \left[ \Ao(x) + \Bo(y) - \left(\Ao(x)\Bo(y) + \Bo(y)\Ao(x)\right)\right]
$$
is the orthogonal projection onto the linear span of the vectors $\{\phii_x,\psi_y\}$. This follows since we have the obvious inclusion $\ran{\Pi(x,y)} \supseteq \spanno{\phii_x,\psi_y}$, and moreover $\Pi(x,y)$ commutes with both projections $\Ao(x)$ and $\Bo(y)$, as one readily sees from the relations $\Ao(x)\Bo(y)\Ao(x) = (1/d)\Ao(x)$ and $\Bo(y)\Ao(x)\Bo(y) = (1/d)\Bo(y)$. Now, the operator $\Co(x,y)$ is a positive multiple of the complementary projection $\id-\Pi(x,y)$, hence not only is $\Co(x,y)$ positive, as claimed, but also $\ran{\Co(x,y)} = \{\phii_x,\psi_y\}^\perp$. The `only if' statement in Thm.~\ref{thm:gamma_d} is then an easy consequence of the next general result.

\begin{proposition}\label{prop:ext_ran}
Suppose $\Ao,\Bo\in\OO{\Omega}$ are any two compatible observables, and let $\Co$ be a joint observable of $\Ao$ and $\Bo$. If $\ran{\Co(x_1,y_1)}\cap\ran{\Co(x_2,y_2)} \neq \{0\}$ for some $(x_1,y_1)\neq (x_2,y_2)$, then $(\Ao,\Bo)$ is not an extreme point of the convex set $\cOO$.
\end{proposition}
\begin{proof}
Let $\eta\in\ran{\Co(x_1,y_1)}\cap\ran{\Co(x_2,y_2)}$ with $\no{\eta} = 1$. Moreover, denote by $P_i$ the orthogonal projection onto $\ran{\Co(x_i,y_i)}$, and let $\lam_i$ be the smallest non-zero eigenvalue of the operator $\Co(x_i,y_i)$. We have
$$
\Co(x_i,y_i) \geq \lam_i P_i \geq \lam_i\kb{\eta}{\eta} \quad \Rightarrow \quad \Co(x_i,y_i) + \varepsilon\kb{\eta}{\eta} \geq 0 \quad \forall \varepsilon\in[-\lam_i\,,\,\lam_i] \,.
$$
Therefore, if $0<\varepsilon\leq\min\{\lam_1,\lam_2\}$, we can define the following two observables $\Co^+,\Co^-\in\OO{\Omega^2}$
$$
\Co^{\pm}(x,y) = \begin{cases}
\displaystyle \Co(x_1,y_1) \pm \varepsilon\kb{\eta}{\eta} & \text{ if } (x,y) = (x_1,y_1) \\
\displaystyle \Co(x_2,y_2) \mp \varepsilon\kb{\eta}{\eta} & \text{ if } (x,y) = (x_2,y_2) \\
\displaystyle \Co(x,y) & \text{ otherwise}
\end{cases}
$$
and their margins $\Ao^\pm$, $\Bo^\pm$. Clearly, $(\Ao,\Bo) = (1/2) \left[ (\Ao^+,\Bo^+) + (\Ao^-,\Bo^-) \right]$, and it is easy to check that $(\Ao^+,\Bo^+) \neq (\Ao^-,\Bo^-)$. This means that $(\Ao,\Bo)$ is not extremal in $\cOO$.
\end{proof}

\begin{proof}[of Thm.~\ref{thm:gamma_d}]
Fix any $x\in\Omega$, and let $y_1,y_2\in\Omega$ with $y_1\neq y_2$. As we already noticed, $\ran{\Co(x,y_1)}\cap\ran{\Co(x,y_2)} = \{\phii_x,\psi_{y_1},\psi_{y_2}\}^\perp$. Since the vectors $\phii_x,\psi_{y_1},\psi_{y_2}$ are linearly independent, this implies $\ran{\Co(x,y_1)}\cap\ran{\Co(x,y_2)}\neq\{0\}$ if and only if $d\geq 4$. Prop.~\ref{prop:ext_ran} then yields the necessity statement in Thm.~\ref{thm:gamma_d}. On the other hand, if $d=3$ we have $\rank{\Co(x,y)} = 1$ for all $x,y$, hence by Prop.~\ref{prop:indep+rk1->ext} and Lemma \ref{lem:K+E} (with $c = -a = -b = 1/[(d-2)(d-1)]$ and $e = 1/[d(d-2)]$\,) the pair of observables $(\Ao_{1/(1-d)},\Bo_{1/(1-d)})$ is an extreme point of $\cOO$ if and only if condition \eqref{it:b_thm_extreme_Gamma_d} of Thm.~\ref{thm:extreme_Gamma_d} hold. In \cite[Prop.~2.1]{Haagerup96}, it was proven that in dimension $d=3$ any two MUB are equivalent to the (unique) Fourier conjugate pair. For such a pair, we already found in Prop.~\ref{prop:sp(Lam)-Fourier} that $-1\notin {\rm sp}(\Lam)$. This concludes the proof of the sufficiency statement in Thm.~\ref{thm:gamma_d}.
\end{proof}

\section{Discussion}\label{sec:discussion}

We have shown that in dimension $d\geq 4$ two uniformly noisy MUB can be extremal in the set of all pairs of compatible observables only if their noise paramaters lie on the part of an ellipse \eqref{eq:ellipse}, and in this case they actually are extremal if and only if $-1$ is an eigenvalue of the Haagerup matrix \eqref{eq:Haagerup_matrix}. We have tested this condition analytically for all Fourier conjugate MUB in arbitrary dimensions, and with computer-assisted techniques for some other inequivalent pairs in low dimensions. Further, we have proven that in dimension $d=7$ and for fixed values of the noise parameters, there exist uniformly noisy pairs of MUB both with and without the property of being extremal in the set of all compatible pairs of observables. The latter fact provides a new geometric manifestation of MUB inequivalence, and shows that pairs of MUB are actually enough to feature concrete differences of the respective noisy versions.

Finally, we have seen that the cases in dimensions $d=2$ and $3$ are special, as in the $d=2$ case no uniformly noisy version of two MUB can yield any extremal pair of compatible observables, while in the $d=3$ case also the exceptional pair $(\Ao_{1/(1-d)}\,,\,\Bo_{1/(1-d)})$ is extremal. 

In principle, the problem of characterizing extremality can be clearly carried over to arbitrary $k$-uples of uniformly noisy versions of MUB. In this case, however, even the $k$-dimensional analogue of the compatibility region depicted in Fig.~\ref{fig:MUBregion} is still unknown. In particular, as a consequence of the results in \cite{DeSkFrBr19} and unlike the $k=2$ case, for fixed dimension $d$ the shape of the compatibility region actually depends on the equivalence class of the $k$-uple of MUB at hand if $k\geq 3$. This indeed already happens for $k=3$ and $d=5$, which is the first known case in which inequivalent $k$-uples of MUB in fact feature different noise robustness. Thus, the extremality problem seems to be quite intractable for $k\geq 3$ and arbitrary equivalence classes of MUB. However, we do not exclude that it may become much simpler and more accessible for specific $k$-uples, like e.g.~the complete sets of $d+1$ MUB obtained by standard methods in odd-prime power dimensions \cite{Ivanovic81,WoFi89,BaBoRoVa02,CaScTo16}.

Still in the case of only two MUB, we leave the existence of extremal uniformly noisy pairs of MUB in even dimensions as a concluding open problem for further investigations.

\bibliographystyle{unsrt}

\end{document}